\documentclass[submission,copyright,noderivs]{eptcs}
\providecommand{\event}{NCMA 2025} 
\usepackage{mathrsfs}
\usepackage{amssymb}
\usepackage{amsmath}
\usepackage{amsthm}
\usepackage{textcomp}
\usepackage{pifont}
\usepackage[usenames,dvipsnames]{xcolor}
\usepackage{refcount}
\usepackage{shuffle}
\usepackage{colortbl}
\usepackage{transparent}

\theoremstyle{plain}
\newtheorem{theorem}{Theorem}
\newtheorem{proposition}[theorem]{Proposition}
\newtheorem{lemma}[theorem]{Lemma}
\newtheorem{corollary}[theorem]{Corollary}

\theoremstyle{definition}
\newtheorem{example}[theorem]{Example}

\newcommand{\lfam}{\mathscr{L}}

\newcommand{\npda}{\textsf{NPDA}}

\newcommand{\dpdawtl}{\textsf{DPDAwtl}}
\newcommand{\nrdpdawtl}{\textsf{nrDPDAwtl}}

\newcommand{\nrdidpdawtl}{\textsf{nrDIDPDAwtl}}
\newcommand{\didpdawtl}{\textsf{DIDPDAwtl}}

\newcommand{\nrnidpdawtl}{\textsf{nrNIDPDAwtl}}
\newcommand{\nidpdawtl}{\textsf{NIDPDAwtl}}

\newcommand{\lba}{\textsf{LBA}}

\newcommand{\leftend}{\mathord{\vartriangleright}}
\newcommand{\rightend}{\mathord{\vartriangleleft}}
\newcommand{\border}{\texttt{\#}}
\newcommand{\dollar}{\texttt{\$}}
\newcommand{\cent}{\texttt{\textcent}}
\newcommand{\valc}{\textrm{VALC}}
\newcommand{\invalc}{\textrm{INVALC}}
\newcommand{\accept}{\texttt{accept}}
\newcommand{\compl}[1]{\overline{#1}}

\newcommand{\cyes}{\textcolor{PineGreen}{\ding{51}}}%
\newcommand{\cnos}{\textcolor{red}{\ding{55}}}%
\newcommand{\cno}{\textcolor{black}{\ding{53}}}%
\newcommand{\copen}{\textcolor{blue}{\textbf{?}}}
\newcommand{\empt}{\mathord{\bot}}

\newcommand{\eoe}{\ifmmode$\hspace*{\fill}$\blacksquare\else\hspace*{\fill}$\blacksquare$\fi\smallskip}

\title{Input-Driven Pushdown Automata with\\ Translucent Input Letters}

\author{Martin Kutrib \and Andreas Malcher \and Matthias Wendlandt
\institute{%
  Institut f\"ur Informatik, Universit\"at Giessen\\
  Arndtstr.~2, 35392 Giessen, Germany}
\email{$\{$kutrib, andreas.malcher, matthias.wendlandt$\}$@informatik.uni-giessen.de}
}
 
\def\titlerunning{Input-Driven Pushdown Automata with Translucent Input Letters}
\def\authorrunning{M.~Kutrib, A.~Malcher, M.~Wendlandt}

\begin{document}

\maketitle

\begin{abstract}
Input-driven pushdown automata with translucent input letters are investigated.
Here, the use of translucent input letters means that the input is processed in several sweeps
and that, depending on the current state of the automaton, some input symbols are visible and can be processed, 
whereas some other symbols are invisible, and may be processed in another sweep.
Additionally, the returning mode as well as the non-returning mode are considered, where in the former
mode a new sweep must start after processing a visible input symbol. Input-driven pushdown
automata differ from traditional pushdown automata by the fact that the actions on the pushdown store
(push, pop, nothing) are dictated by the input symbols.
We obtain the result that the input-driven nondeterministic model is computationally stronger than the deterministic model
both in the returning mode and in the non-returning mode, whereas it is known that the deterministic and the nondeterministic
model are equivalent for input-driven pushdown automata without translucency. It also turns out that the non-returning
model is computationally stronger than the returning model both in the deterministic and nondeterministic case.
Furthermore, we investigate the closure properties of the language families introduced under the Boolean operations. 
We obtain a complete picture
in the deterministic case, whereas in the nondeterministic case the language families are shown to be not closed under
complementation. Finally, we look at decidability questions and obtain the non-semidecidability of the
questions of universality, inclusion, equivalence, and regularity in the nondeterministic case. 
\end{abstract}

\section{Introduction}

The usual way of processing input on language recognition devices is by reading input strings from left to right, one symbol at a time, 
and finally providing an accept or reject decision when arriving at the end of the input.
This ``standard'' input mode has been extended in the literature in several directions. For example, 
two-way motion of the input head, stationary moves of the input head, rotating the input head, or restarting modes 
have been considered for a wide variety of machines (the reader is referred to, e.g., 
\cite{bensch:2009:irdnfa,Hopcroft:1979:itatlc:book,Jancar:1995:ra:proc,otto:2025:book} for an overview of these and other models).
In all these cases, the extensions have consequences on the computational and descriptional power of the machines.
Thus, the way of processing the input can be considered as a computational resource that then can be used to tune 
computational and descriptional power.

Of particular interest in recent literature have been \emph{discontinuous} ways of input processing, where one of them
is the ``jumping'' paradigm which means that jumping to any position inside the input string is allowed at any move. 
This paradigm has been investigated for finite automata in~\cite{meduna:2012:jfa}, where it is shown that this discontinuous
input processing may increase the computational power since some non-context-free languages can be accepted. On the other hand,
the discontinuous input mode may also limit the computational power since some regular languages cannot be accepted.
A restricted variant, called ``right one-way jumping'' automata, has been considered in~\cite{beier:2022:nrowjfa,chigahara:2016:owjfa}. 
Another way to discontinuously process the input is to use \emph{translucent letters}. 
This concept  has been introduced by Nagy and Otto in~\cite{nagy:2011:fsawtl} for deterministic and nondeterministic finite automata
and the basic idea for translucent devices is to provide a translucency function that defines in which states which letters of the input are translucent. 
At each move, the device skips (by looking through) the translucent portion of the input, from the
current input head position up to the first non-translucent letter (thus realizing a jump).
After processing the non-translucent symbol, in the \emph{returning mode} the
input head returns to the left end of the input while, in the \emph{non-returning mode}, the input head continues to process the input 
according to its updated current state and the corresponding translucent symbols. In both modes, the input head returns to the left end when the right
end of the input is reached. 

Deterministic and nondeterministic finite automata with translucent letters have deeply been investigated in the literature 
(see, e.g.,~\cite{mraz:2023:nrfawtl,Nagy:2012:ocdsosdrawws,otto:2023:asoawtl:proc}). However, many questions are still open. 
Some recently studied variants are finite automata with translucent words~\cite{Nagy:2024:fawsotw}, finite automata with translucent letters and
two-way motion of the input head~\cite{kutrib:2025:twfawtil:proc}, and returning and non-returning deterministic pushdown automata with translucent 
letters~\cite{kutrib:2024:dpdawtl:proc,Kutrib:2024:opolabdpawtil}.
It should be noted that returning pushdown automata with translucent letters have been studied by Nagy and Otto 
in~\cite{nagy:2011:cdsosdragbaeps,Nagy:2013:dpcdsosdra} in terms of certain cooperating distributed systems of restarting automata with 
additional pushdown store. However, in their model $\lambda$-transitions are not allowed and acceptance is defined by empty pushdown.
In~\cite{Otto:2015:ovptl} the paradigm of \emph{input-driven} languages was imposed to certain cooperating distributed systems of restarting automata 
with additional pushdown store with regard to characterizing certain trace languages. 
In input-driven pushdown automata~\cite{braunmuehl:1985:idlrlog-coll,mehlhorn:1980:pmradcflr} the next action on the pushdown store (push, pop, nothing) 
is solely governed by the input symbol and to this end the input alphabet is split into three subsets. 
Input-driven pushdown automata possess nice features (see, e.g., the survey given in~\cite{okhotin:2014:cidpda}) such as 
the equivalence of nondeterministic and deterministic
models, the positive decidability of the inclusion problem and the positive closure under union, intersection, complementation, concatenation, and iteration.
It should be noted that the positive decidability result as well as the positive closures only hold if both given automata have a compatible splitting of
their input alphabets. In this paper, we study input-driven deterministic and nondeterministic pushdown automata with translucent letters
both in the returning and non-returning mode and, therefore, extend the results from~\cite{Otto:2015:ovptl}. 
We also study the closure properties of the four corresponding language families under the Boolean operations 
as well as the status of their decidability questions. 

The paper is organized as follows. After giving the necessary definitions and two illustrating examples in the next section, 
we study in Section~\ref{sec:det-vs-ndet} the impact of nondeterminism in our models and it turns out
that the nondeterministic model is computationally stronger in the returning mode as well as in the
non-returning mode. In Section~\ref{sect:ret} we compare returning and non-returning models. 
We yield proper inclusions of the returning language families in the non-returning language families
in the deterministic as well as in the nondeterministic case. Moreover, the combination deterministic and
non-returning versus the combination nondeterministic and returning leads to incomparability results.
The closure under the Boolean operations is studied in Section~\ref{sec:closure} and we obtain the
closure under complementation, but non-closure under union and intersection in the returning and
non-returning deterministic case. In the nondeterministic case, we obtain non-closure under complementation
for both variants and non-closure under intersection (even with regular languages) in the returning case.
Finally, we consider the usually studied decidability questions in Section~\ref{sect:deci}.
In the returning case we have the decidability of the emptiness problem as well as of the finiteness problem both
in the nondeterministic and deterministic case. In addition, universality is decidable in the deterministic
case. On the other hand, we get the non-semidecidability of the questions of universality, equivalence, inclusion, 
and regularity in case of input-driven nondeterministic pushdown automata with translucent letters working in the 
returning as well as in the non-returning mode.
This extends some results partially known for nondeterministic finite automata with translucent input letters
working in the returning mode (see~\cite{Nagy:2012:ocdsosdrawws}).

\section{Definitions and Preliminaries}\label{sec:prelim}

We denote by $\Sigma^*$ the set of all words on the finite alphabet $\Sigma$, including the 
\emph{empty word} $\lambda$, and let
$\Sigma^+ = \Sigma^* \setminus \{\lambda\}$. For any word $w\in\Sigma^*$, we let 
$|w|$ denote its \emph{length}, $w^R$ its \emph{reversal}, and $|w|_a$ the \emph{number of occurrences} of the 
symbol $a\in\Sigma$ in $w$. 
We use $\subseteq$ for \emph{inclusions}, and $\subset$ for \emph{proper inclusion}. 
Given a set $S$, we denote by $2^S$ its \emph{power set}, and by $|S|$ its \emph{cardinality}. 
We write $S_x$ to denote the set $S\cup\{x\}$, for a given element $x\not\in S$.
A \emph{language} on  $\Sigma$ is any subset $L \subseteq \Sigma^*$. The \emph{complement} of $L$
is the language $\compl{L}=\Sigma^* \setminus L$, its reversal is $L^R=\{\,w^R \mid w\in L\,\}$.
The \emph{shuffle} $L_1 \shuffle L_2$ of two languages $L_1, L_2 \subseteq \Sigma^*$ is defined as
$L_1 \shuffle L_2 = \{\, x_1y_1x_2y_2\cdots x_ny_n \mid x_1x_2 \cdots x_n \in L_1, y_1y_2 \cdots y_n \in L_2
\mbox{ with } x_i,y_i \in \Sigma^* \mbox{ and } 1\le i\le n\,\}$.
Let $\Sigma=\{a_1,a_2,\ldots,a_n\}$ and $\Psi : \Sigma^* \rightarrow \mathbb{N}_0^n$ be a mapping
such that $\Psi(w)=(|w|_{a_1},|w|_{a_2},\ldots,|w|_{a_n})$. Let $L \subseteq \Sigma^*$ be a language.
Then, $\Psi(L)=\{ \Psi(w) \mid w \in L\}$ is the \emph{Parikh image} of~$L$. We say that two
languages $L_1,L_2 \subseteq \Sigma^*$ are \emph{letter-equivalent} if $\Psi(L_1)=\Psi(L_2)$.
Two language families $\mathscr{L}_1$ and~$\mathscr{L}_2$
are said to be \emph{incomparable} whenever~$\mathscr{L}_1$ is not a subset of~$\mathscr{L}_2$
and vice versa.

A classical pushdown automaton is called input-driven
if the current input symbol defines the next action on the pushdown store, that is,
pushing a symbol onto the pushdown store,
popping a symbol from the pushdown store, or changing the  
state without modifying the pushdown store. 
To this end, the input alphabet~$\Sigma$ is partitioned
into the sets~$\Sigma_D$,~$\Sigma_R$, and~$\Sigma_N$,
that control the actions push~($D$), pop~($R$), and state change only~($N$).

Input-driven pushdown automata with translucent input letters are extensions of 
input-driven pushdown automata that do not necessarily have to read the
current input symbol.
Instead, depending on the current state of such devices, some of the input
letters may be translucent (invisible). Accordingly, an input-driven pushdown automaton
with translucent input letters either reads and processes 
(by deleting, if not the endmarker) the first visible input letter.

\begin{sloppypar}
Formally, an \emph{input-driven pushdown automaton with translucent input
letters $(\nidpdawtl)$} is a system 
$M=\langle Q,\Sigma, \Gamma, q_0, \rightend, \bot, \tau, \delta_D, \delta_R, \delta_N\rangle$, 
where
$Q$ is the finite set of \emph{internal states},
$\Sigma$ is the finite set of \emph{input symbols} partitioned into the sets 
$\Sigma_D$, $\Sigma_R$, and~$\Sigma_N$, with
$\Sigma \cap Q=\emptyset$,
$\Gamma$ is the finite set of \emph{pushdown symbols},
$q_0 \in Q$ is the \emph{initial state},
$\rightend\notin \Sigma$ is the {endmarker},
$\bot \notin \Gamma$ is the \emph{bottom-of-pushdown symbol},
\mbox{$\tau\colon Q\to 2^\Sigma$} is the {translucency mapping}, and
$\delta_D$ is the partial transition function mapping 
   \mbox{$Q \times \Sigma_D\times (\Gamma \cup\{\empt\})$} to $2^{Q \times \Gamma}\cup\{\accept\}$,
$\delta_R$ is the partial transition function mapping 
  $Q \times \Sigma_R \times (\Gamma \cup\{\empt\})$ to~$2^Q\cup\{\accept\}$, and
$\delta_N$ is the partial transition function mapping 
  $Q \times (\Sigma_N\cup\{\rightend\}) \times (\Gamma \cup\{\empt\})$ to~$2^Q\cup\{\accept\}$.
\end{sloppypar}

The translucency mapping $\tau$ bears the following meaning:
for any state $q\in Q$, the letters from the set $\tau(q)$ are 
\emph{translucent (invisible) for $q$}, that is, whenever in $q$, the automaton $M$ does
not see these letters (or equivalently, $M$ sees through such letters).

A \emph{configuration} of $M$ is a pair $(q w\rightend,\gamma)$ or $\accept$, 
where $q \in Q$ is the current state, $w\in \Sigma^*$ is the {part of the
input} left to be processed, and
\mbox{$\gamma \in \Gamma^*\bot$} denotes the current pushdown content, the leftmost symbol being
 the top of the pushdown store.
The \emph{initial configuration} for an input~$w$ is set to $(q_0 w\rightend,\empt)$.

Being in some configuration $(q w\rightend, z\gamma)$ with
$z\in\Gamma\cup\{\empt\}$ and $z\gamma\in \Gamma^*\empt$, first~$M$ determines the next symbol to scan.
Precisely, if~$w\rightend=xy\rightend$ with $x\in \tau(q)^*$, $y\in\Sigma^*$,
and $a\notin \tau(q)$ is the first symbol of $y\rightend$,
then $M$ processes~$a$.
One step from a configuration to its {successor configuration} is denoted
by~$\vdash$.

Let $q,q'\in Q$, $x\in \tau(q)^*$, $a\notin \tau(q)$, $y\in\Sigma^*$, and
$z\in \Gamma$, $\gamma\in \Gamma^*\empt$. 
We set
\begin{enumerate}
\item 
  $(q xay \rightend,z\gamma) \vdash (q' xy\rightend,z'z\gamma)$, 
  if $a \in \Sigma_D$ and $(q',z') \in \delta_D(q,a,z)$, 
\item 
  $(q xay \rightend,\empt) \vdash (q' xy\rightend,z'\empt)$, 
  if $a \in \Sigma_D$ and $(q',z') \in \delta_D(q,a,\empt)$, 
\item 
  $(q xay \rightend,z\gamma) \vdash (q' xy\rightend,\gamma)$, 
  if $a \in \Sigma_R$ and $q'\in \delta_R(q,a,z)$, 
\item 
  $(q xay \rightend,\empt) \vdash (q' xy\rightend,\empt)$, 
  if $a \in \Sigma_R$ and $q'\in \delta_R(q,a,\empt)$, 
\item 
  $(q xay \rightend,z\gamma) \vdash (q' xy\rightend, z\gamma)$, 
  if $a \in \Sigma_N$ and $q'\in \delta_N(q,a,z)$, 
\item 
  $(q xay \rightend,\empt) \vdash (q' xy\rightend,\empt)$, 
  if $a \in \Sigma_N$ and $q'\in \delta_N(q,a,\empt)$, 
\item 
  $(q x \rightend,z\gamma) \vdash (q' x\rightend, z\gamma)$, 
  if $q'\in \delta_N(q,\rightend,z)$, 
\item 
  $(q x \rightend,\empt) \vdash (q' x\rightend,\empt)$, 
  if $q'\in \delta_N(q,\rightend,\empt)$. 
\end{enumerate}
In addition, whenever, the transition function yields $\accept$ for the current
configuration, the successor configuration is $\accept$. 

So, on the endmarker only $\delta_N$ is defined.
Whenever the pushdown store is empty, the successor configuration is computed
by the transition functions with the special bottom-of-pushdown symbol~$\empt$
which is never removed from the pushdown.
As usual, we define the reflexive and transitive closure of $\vdash$ by
$\vdash^*$.

\begin{sloppypar}
An input-driven pushdown automaton with translucent letters is said
to be deterministic ($\didpdawtl$)
if $|\delta_x(q,a,z)|\leq 1$ for $x\in\{D,N,R\}$ and all $q\in Q$, $a\in\Sigma_x$, and
$z\in \Gamma \cup\{\bot\}$.
\end{sloppypar}

A word $w$ is accepted by $M$ if there is a computation on input $w$ that ends
with $\accept$.
The \emph{language accepted} by $M$ is
$L(M)=\{\, w\in \Sigma^* \mid w \text{ is accepted by } M\,\}$.
In general, the family of all languages accepted by automata of 
some type~$\mathsf{X}$ will be denoted by $\lfam(\mathsf{X})$.

Some properties of language families implied by classes of 
input-driven pushdown automata may depend on whether 
all automata involved share the same partition of the input
alphabet. For easier writing, we call the partition of an input alphabet 
a \emph{signature}.

In order to clarify these notions, we continue with an example.

\begin{example}\label{exa:rep}
Let $\Sigma=\{a,b,\border\}$ be an alphabet. The language
$L_{\textsl{rep}}\subseteq \Sigma^*$ is defined as
$
L_{\textsl{rep}}= \{\, b^n (\border b^n)^k \shuffle a^n \mid n\geq 1, k\geq 0\,\}
$. Its complement $\compl{L}_{\textsl{rep}}$ belongs to the family $\lfam(\nidpdawtl)$.
It can be represented as union $L\cup L'$ where $L'$ is the complement of the regular language
$\{\, b^+ (\border b^+)^k \shuffle a^+ \mid k\geq 0\,\}$ with respect to
$\Sigma$, and
$$
L=\{\, b^{n_1} \border b^{n_2} \border \cdots \border b^{n_k} \shuffle a^n \mid k\geq 0,
n,n_i\geq 1 \text{ for } 1\leq i\leq k, \text{ and there exists } 1\leq i\leq
k \text{ such that } n_i\neq n\,\}.
$$

An $\nidpdawtl$ accepting $\compl{L}_{\textsl{rep}}$ can initially guess
whether it wants to accept $L$ or $L'$. Since $L'$ is regular it is accepted
by some $\nidpdawtl$ that does not utilize its pushdow store. So, it does not
care about what actually happens on the pushdown store. This means that it may
have an arbitrary signature.

Now, we construct an $\nidpdawtl$
$M=\langle Q,\Sigma, \Gamma, q_0, \rightend, \bot, \tau, \delta_D, \delta_R,
\delta_N\rangle$ accepting $L$ as follows.
\begin{itemize}
\item $Q=\{q_0,q_1,p,r\}$,

\item $\Sigma_D =\{b\}$, $\Sigma_R =\{a\}$, $\Sigma_N =\{\border\}$,
\item $\Gamma=\{\bullet,B\}$, 

\item 
$\tau(q_0)=\tau(q_1)=\tau(p)=\{a\}$,\quad $\tau(r)=\{\border,b\}$.
\end{itemize}

In a first phase, $M$ reads the input symbols $b$ and $\border$ from left to
right with symbols $a$ translucent. At the beginning and whenever a $\border$
is read, $M$ guesses whether the length of the next $b$-block has to be
matched with the number of $a$'s in the input. If not, $M$ scans the $b$-block
and pushes $\bullet$'s in state $q_1$. If yes, $M$ enters state $p$ and scans
the $b$-block as well but now pushing $B$'s. If $M$ never guesses yes, the
computation is rejected on the endmarker. Otherwise, on reading the next
$\border$ or the endmarker, $M$ enters state $r$. In this situation, the
number of $B$'s on top of the pushdown corresponds to the length of the
$b$-block guessed to be matched. Finally, for state $r$ the symbols $b$ and
$\border$ are translucent. Now $M$ reads the $a$'s from left to right. For
each $a$ read, one $B$ is popped. If and only if there are more $a$'s than $B$'s or vice
versa then $M$ accepts.
So, we set:
\begin{center}
\renewcommand{\arraystretch}{1.1}\tabcolsep2pt
\begin{tabular}[t]{rccl}
(1) &  $\delta_D(q_0,b,\bot)$ &=& $\{(q_1,\bullet),(p,B)\}$,\\
(2) &  $\delta_D(q_1,b,\bullet)$ &=& $\{(q_1,\bullet)\}$,\\
(3) &  $\delta_N(q_1,\border,\bullet)$ &=& $\{q_1,p\}$,\\
(4) &  $\delta_D(p,b,\bullet)$ &=& $\{(p,B)\}$,\\
(5) &  $\delta_D(p,b,B)$ &=& $\{(p,B)\}$,\\
(6) &  $\delta_N(p,\border,B)$ &=& $\{r\}$,\\
\end{tabular}
\qquad\qquad
\begin{tabular}[t]{rccl}
(7) &  $\delta_N(p,\rightend,B)$ &=& $\{r\}$,\\
(8) &  $\delta_R(r,a,B)$ &=& $\{r\}$,\\
(9) &  $\delta_R(r,a,\bullet)$ &=& $\accept$,\\
(10) &  $\delta_R(r,a,\bot)$ &=& $\accept$,\\
(11) &  $\delta_R(r,\rightend,B)$ &=& $\accept$.\\
\end{tabular}
\end{center}
\eoe
\end{example}

Deterministic pushdown automata with translucent letters have been defined in~\cite{kutrib:2024:dpdawtl:proc}
also for the \emph{non-returning} mode. In this mode, after a visible letter is processed, 
the input head does not return to the left end of the input, but it continues reading
from the position of the visible letter just processed. Whenever the endmarker is reached and the 
transition on the endmarker yields a new state, the computation is continued 
in this state, with the input head placed at the left end of the remaining input.
Such deterministic pushdown automata with translucent letters working in the
non-returning mode generalize non-returning finite state automata 
with translucent letters~\cite{mraz:2023:nrfawtl}. Here, we also consider
\emph{input-driven pushdown automata with translucent letters working in the
non-returning mode} ($\nrnidpdawtl,\nrdidpdawtl$).

Let
$M=\langle Q,\Sigma, \Gamma, q_0, \rightend, \bot, \tau, \delta_D, \delta_R, \delta_N\rangle$, 
be an $\nrnidpdawtl$. Now, a configuration of~$M$ is a pair $(uqw\rightend,\gamma)$ or $\accept$, 
where $q \in Q$ is {the current state}, $uw\in \Sigma^*$ is the {remaining part of the
input} with $w$ being to the right and $u$ to the left of the input head, and
$\gamma \in \Gamma^*\bot$ is the current pushdown content.
The successor configuration yielded by $\vdash$ is now specified
as follows.
Let
Let $q,q'\in Q$, $x\in \tau(q)^*$, $a\notin \tau(q)$, $u,y\in\Sigma^*$, and
$z\in \Gamma$, $\gamma\in \Gamma^*\empt$. Then:
\begin{enumerate}
\item 
  $(uq xay \rightend,z\gamma) \vdash (ux q' y\rightend,z'z\gamma)$, 
  if $a \in \Sigma_D$ and $(q',z') \in \delta_D(q,a,z)$, 
\item 
  $(u q xay \rightend,\empt) \vdash (ux q' y\rightend,z'\empt)$, 
  if $a \in \Sigma_D$ and $(q',z') \in \delta_D(q,a,\empt)$, 
\item 
  $(uq xay \rightend,z\gamma) \vdash (uxq' y\rightend,\gamma)$, 
  if $a \in \Sigma_R$ and $q'\in \delta_R(q,a,z)$, 
\item 
  $(u q xay \rightend,\empt) \vdash (uxq' y\rightend,\empt)$, 
  if $a \in \Sigma_R$ and $q'\in \delta_R(q,a,\empt)$, 
\item 
  $(uq xay \rightend,z\gamma) \vdash (uxq' y\rightend, z\gamma)$, 
  if $a \in \Sigma_N$ and $q'\in \delta_N(q,a,z)$, 
\item 
  $(uq xay \rightend,\empt) \vdash (uxq' y\rightend,\empt)$, 
  if $a \in \Sigma_N$ and $q'\in \delta_N(q,a,\empt)$, 
\item 
  $(u q x \rightend,z\gamma) \vdash (q' ux\rightend, z\gamma)$, 
  if $q'\in \delta_N(q,\rightend,z)$, 
\item 
  $(u q x \rightend,\empt) \vdash (q' ux\rightend,\empt)$, 
  if $q'\in \delta_N(q,\rightend,\empt)$. 
\end{enumerate}
In addition, whenever, the transition function yields $\accept$ for the current
configuration, the successor configuration is $\accept$. 

\begin{sloppypar}
The accepted language $L(M)$ can be easily defined.
Sometimes, we will be saying that an $\nrnidpdawtl$ performs \emph{sweeps}, where a sweep 
is a sequence of transitions that starts with the input head at the left end of the
(remaining) input and ends after the next (if any) return move on the endmarker.
Let us give an intuition on how an $\nrnidpdawtl$ works by the following example.
\end{sloppypar}

\begin{example}\label{exa:intro-nonret}
We consider two languages over the alphabet $\{a,b\}$ together with its overlined variant $\{\bar{a},\bar{b}\}$
and its doubly-overlined variant $\{\bar{\bar{a}},\bar{\bar{b}}\}$. In addition, we consider two mappings
$h_1$ and $h_2$ such that
$h_1(a)=\bar{a}$, $h_1(b)=\bar{b}$, $h_2(a)=\bar{\bar{a}}$, and $h_2(b)=\bar{\bar{b}}$.
Then, the languages $L_1$ and $L_2$ are defined as follows.
\begin{eqnarray*}
 L_1 &=& \{\, w v h_2({w^R})\mid w\in\{a,b\}^+, v\in\{\bar{a},\bar{b}\}^+\,\},\\
 L_2 &=& \{\, w h_1({w^R}) v \mid w\in\{a,b\}^+, v\in\{\bar{\bar{a}},\bar{\bar{b}}\}^+\}.
\end{eqnarray*}
We will show that the union $L=L_1 \cup L_2$ is accepted by a 
nondeterministic input-driven pushdown automaton with translucent input letters in the non-returning mode. 
We define an $\nrnidpdawtl$ $M=\langle Q,\Sigma, \Gamma, q_0, \rightend, \bot, \tau, \delta_D, \delta_R, \delta_N\rangle$ 
accepting~$L$ as follows.
\begin{itemize}
\item $Q=\{q_0, q_1, q_2,q_3, q_4, q_5\}$,
\item $\Sigma_D=\{ a,b\}$, $\Sigma_R=\{ \bar{a},\bar{b},\bar{\bar{a}}, \bar{\bar{b}}\}$, $\Sigma_N=\emptyset$,
\item $\Gamma=\{A,B\}$, 
\item 
$\tau(q_0)=\tau(q_2)=\tau(q_3)=\tau(q_4)=\tau(q_5)=\emptyset$,\quad $\tau(q_1)=\{\bar{a}, \bar{b}\}$.
\end{itemize}

In a first phase, $M$ reads input symbols $a$ or $b$ and pushes corresponding symbol $A$ or $B$ onto
the pushdown store. At any time step, $M$ decides nondeterministically whether it remains in this phase 
by remaining in state $q_0$ or
whether it starts to test whether the input belongs to the first set of the union $L_1$ 
by entering state $q_1$ or to the second set $L_2$ by entering state $q_3$.

To test whether the input belongs to $L_1$, $M$ enters state~$q_1$ and all symbols $\bar{a}$ or $\bar{b}$ become translucent.
Then, $M$ starts to match all remaining doubly-overlined input symbols against the pushdown store.
To ensure the correct format of the remaining input, $M$ changes its state whenever the first doubly-overlined 
symbol is read. If $M$ reaches the endmarker and the pushdown store is empty,
the accepting state is entered. This is realized in rules (3)--(7). 

To test whether the input belongs to $L_2$, $M$ enters state~$q_3$ and no symbols are translucent.
Then, $M$ starts to match all following single-overlined input symbols against the pushdown store.
Once all such symbols are read and the pushdown store is empty, $M$ proceeds to read the doubly-overlined
symbols while popping from the (already empty) pushdown store. 
To ensure the correct format of the input, $M$ changes its state whenever the first overlined and the 
first doubly-overlined symbol is read.
If $M$ encounters the endmarker and the pushdown store is empty, $M$ enters the accepting state.
This is realized in rules (8)--(16). 

The transition functions are defined as follows for $Z \in \{ \bot, A, B\}$.
\begin{center}
\renewcommand{\arraystretch}{1.1}\tabcolsep2pt
\begin{tabular}[t]{rccl}
(1) &  $\delta_D(q_0,a,Z)$ &=& $\{(q_0,A),(q_1,A),(q_3,A)\}$,\\
(2) &  $\delta_D(q_0,b,Z)$ &=& $\{(q_0,B),(q_1,B),(q_3,B)\}$,\\
(3) &  $\delta_R(q_1,\bar{\bar{a}},A)$ &=& $\{q_2\}$,\\
(4) &  $\delta_R(q_1,\bar{\bar{b}},B)$ &=& $\{q_2\}$,\\
(5) &  $\delta_R(q_2,\bar{\bar{a}},A)$ &=& $\{q_2\}$,\\
(6) &  $\delta_R(q_2,\bar{\bar{b}},B)$ &=& $\{q_2\}$,\\
(7) &  $\delta_N(q_2,\rightend,\bot)$ &=& $\accept$,\\
(8) &  $\delta_R(q_3,\bar{a},A)$ &=& $\{q_4\}$,\\
\end{tabular}
\qquad\qquad
\begin{tabular}[t]{rccl}
(9) &  $\delta_R(q_3,\bar{b},B)$ &=& $\{q_4\}$,\\
(10) &  $\delta_R(q_4,\bar{a},A)$ &=& $\{q_4\}$,\\
(11) &  $\delta_R(q_4,\bar{b},B)$ &=& $\{q_4\}$,\\
(12) &  $\delta_R(q_4,\bar{\bar{a}},\bot)$ &=& $\{q_5\}$,\\
(13) &  $\delta_R(q_4,\bar{\bar{b}},\bot)$ &=& $\{q_5\}$,\\
(14) &  $\delta_R(q_5,\bar{\bar{a}},\bot)$ &=& $\{q_5\}$,\\
(15) &  $\delta_R(q_5,\bar{\bar{b}},\bot)$ &=& $\{q_5\}$,\\
(16) &  $\delta_N(q_5,\rightend,\bot)$ &=& $\accept$.\\
\end{tabular}
\end{center}

Clearly, if there is any error in the format of the input or the part of input compared to the pushdown store does
not match, the transition function of~$M$ is not defined and, therefore, the input is rejected.

We would like to note that~$M$ uses its translucency to ``overlook'' the overlined part of an input belonging to~$L_1$,
since otherwise the pushdown store may be inadvertently emptied and could not be matched with the doubly-overlined input~$w^R$.
Moreover, $M$ uses the non-returning mode to ensure that the input is correctly formatted, that is, symbols from $\{a,b\}^+$
are followed by symbols from $\{\bar{a},\bar{b}\}^+$ which are in turn followed by symbols from $\{\bar{\bar{a}},\bar{\bar{b}}\}^+$.
\eoe
\end{example}

As is known for finite automata and regular pushdown
automata~\cite{kutrib:2024:dpdawtl:proc},
also for input-driven pushdown automata with translucent letters it holds that
any automaton $M$ working in the returning mode can be simulated by some
automaton $M'$ working in the non-returning mode, where both share the same signature and the
same transclucency mapping (for the states of $M$).
The construction of $M'$ for a given~$M$, roughly speaking, works as follows:
at each step, $M'$ simulates the step of $M$, followed by a new step which brings the input head to the leftmost 
input symbol. To this end, let $Q'$ be a primed copy of $Q$.
The transition function~$\delta_i$, for $i\in \{D,R,N\}$, is modified
to~$\delta'_i$ so that the state $q'\in Q'$ is entered if and only if $M$ enters the state $q\in Q$.
The translucency mapping $\tau$ is extended to $\tau'$ by adding
$\tau'(q')=\Sigma$, for all $q'\in Q'$. This clearly implies that, whenever in any state
from $Q'$, $M'$ sees the endmarker.

Finally, $\delta'_N$ is extended by $\delta'_N(q',\rightend,z)=\{q\}$, for all $q'\in Q'$ and all $z\in
\Gamma\cup\{\bot\}$, thus bringing the input head to the leftmost position. One may easily verify that
$M'$ accepts if and only $M$ accepts.

\section{Determinism versus Nondeterminism}\label{sec:det-vs-ndet}

It is well-known that for input-driven pushdown automata, deterministic
devices are as powerful as their nondeterministic counterparts. This contrasts
with the general case of pushdown automata, where the deterministic variant is
strictly weaker than the nondeterministic one. 
In the following, we examine the situation for deterministic and
nondeterministic input-driven pushdown automata with translucent input
letters. It turns out that the family of languages accepted by deterministic
input-driven pushdown automata with translucent input letters is a proper
subset of the family of languages accepted by their nondeterministic
counterparts, both in the returning and in the non-returning case.

\begin{theorem}\label{theo:det-in-ndet}
 The family $\lfam(\didpdawtl)$ is properly included in the family $\lfam(\nidpdawtl)$
 and the family $\lfam(\nrdidpdawtl)$ is properly included in the family $\lfam(\nrnidpdawtl)$.
\end{theorem}

\begin{proof}
We use the union $L=L_1\cup L_2$ as witness language for the properness of the
inclusions, where
$
 L_1= \{\, b^n\border b^m \shuffle a^n\mid m,n\geq 1\,\}
$
and
$
 L_2 = \{\, b^m\border b^n \shuffle a^n \mid m,n\geq 1\,\}
$.

Similarly as in Example~\ref{exa:rep}, two $\didpdawtl$'s can be constructed that
accept $L_1$ respectively $L_2$, and an $\nidpdawtl$ can be constructed that
accepts $L=L_1\cup L_2$.

It remains to be shown that $L$ is not accepted by any
$\nrdidpdawtl$. Assume in contrast to the assertion that $L$ is accepted by
some $\nrdidpdawtl$
$M=\langle Q,\Sigma, \Gamma, q_0, \rightend, \bot, \tau, \delta_D, \delta_R, \delta_N\rangle$.
We consider accepting computations on inputs of the form 
$a^n b^k \border b^\ell$ where $n,k,\ell$ are large enough. 

A basic observation is that $M$ cannot access the $b$-block following the
$\border$ unless the $\border$ is read or there are no more $b$'s in front of
the $\border$.

First we claim that $M$ cannot read the symbol $\border$ and
return to the left of the input before one or both $b$-blocks are
read entirely. 

Assume in contrast to the claim that $M$ is in some configuration
$
(a^{m_1} q_1 b^{k_1}\border b^{\ell_1} \rightend,\gamma_1)
$
such that \mbox{$b\in\tau(q_1)$,} and from the successor configuration
$
(a^{m_1} b^{k_1} q_2 b^{\ell_1} \rightend,\gamma_2)
$
it reads some further $b$'s and then jumps to the endmarker reaching a configuration
$
(q_3 a^{m_1} b^{k_1} b^{\ell_2} \rightend, \gamma_3).
$
If neither $b^{k_1}$ nor $b^{\ell_2}$ is empty then
also the inputs
$a^n b^{k+1} \border b^{\ell-1}$ and $a^n b^{k-1} \border b^{\ell+1}$
would be accepted but at least one of them does not belong to $L$. 
The contradiction shows the claim.

Next, we have to consider four cases.

\textbf{Case 1:} The first case is that
neither $a\in \Sigma_D$ and $b\in\Sigma_R$ nor
$a\in \Sigma_R$ and $b\in\Sigma_D$.
Essentially, this means that $M$ does not use its pushdown.

A single sweep of $M$ is analyzed. If~$M$ reads more than $2|Q|$ consecutive 
symbols $a$ (respectively~$b$), it
must eventually enter a loop of, say $c$, states. Since the pushdown is not
used, we can increase the length of the $a$-block (respectively $b$-block) 
by $c$ without changing the overall result of the computation, a contradiction. 
Hence, in this case, in any sweep $M$ cannot enter a loop while reading $a$'s or
$b$'s. Therefore, $M$ must perform multiple sweeps without loops. For the
states in which the sweeps start there are at most $|Q|$ possibilities. Due to
the deterministic behavior, eventually $M$ will run through loops of sweeps with respect
to the starting state of the loop.
Let $x_0$ be the number of $a$'s read before this sweep loop and let $x_1$ be
the number of $a$'s consumed in one sweep loop. Similarly, let $y_0$ and
$y_1$ be the numbers of $b$'s consumed.
Now we consider the input 
$a^{x_0+x_1} b^{y_0+y_1} \border b^\ell$
where $\ell > 2|Q|$. 
Then, after the first sweep loop, the $a$-block and the first $b$-block 
are entirely read. If $x_0+x_1\neq y_0+y_1$,
$M$ must verify whether $x_0+x_1 = \ell$  
To do so, $M$ needs to read the second block of $b$'s 
-- but without using the pushdown store and without any leftover
$a$'s. Given that $\ell > 2|Q|$, $M$ must enter a state loop while reading the
second $b$-block, say of length $c'$. Hence, if $M$ accepts the input for
some $\ell$, it accepts the input with $\ell + c'$, which contradicts the
definition of $L$. 
If $x_0+x_1= y_0+y_1$, almost the same argument holds for the input 
$a^{x_0+x_1} b^{y_0+y_1-1} \border b^\ell$.

\textbf{Case 2:}
The second case is that $a\in \Sigma_D$ and $b\in \Sigma_R$. Moreover, $M$ does not enter
loops with respect to the current state and the topmost pushdown symbol.

In this case, in one sweep, $M$ reads at most $|Q|(|\Gamma|+1)$ symbols of 
the $a$-block and the first $b$-block (as long as these blocks are non-empty).
We argue similarly as in Case 1. Due to
the deterministic behavior, eventually $M$ will run through loops of sweeps with respect
to the starting state of the loop.
Let $x_0$ be the number of $a$'s read before this sweep loop and let $x_1$ be
the number of $a$'s consumed in one sweep loop. Similarly, let $y_0$ and
$y_1$ be the numbers of $b$'s consumed.

If $x_0+x_1> y_0+y_1$, we consider the input 
$a^{x_0+x_1} b^{x_0+x_1+y_0+y_1} \border b^\ell$
where $\ell > 2|Q|$. 
Then, after the first sweep loop, all $a$'s are read, 
and the remaining input is $b^{x_0+x_1} \border b^\ell$. 
Moreover, there are $x_0+x_1 - y_0-y_1$ symbols left in
the pushdown store. 
If $x_0+x_1> y_0+y_1$, next $M$ will empty its pushdown store 
after reading $x_0+x_1 - y_0-y_1$ many $b$'s. 
Afterward, $M$ can only rely on its finite
set of states. Hence, after reading at most $|Q|$ additional
symbols $b$, $M$ enters a state loop of some length $c''$. 
Consequently, if $M$ accepts the input for some $\ell$ it must also accept the input
with $\ell + c''$, which contradicts the
definition of~$L$.

The same reasoning applies 
for the input $a^{x_0+x_1} b^{y_0+y_1} \border b^\ell$,
where $\ell > 2|Q|$, if $x_0+x_1 < y_0+y_1$, and
for the input $a^{x_0+x_1} b^{y_0+y_1-1} \border b^\ell$,
where $\ell > 2|Q|$, if $x_0+x_1 = y_0+y_1$.

\textbf{Case 3:}
The third case is that $a\in \Sigma_D$ and $b\in \Sigma_R$ and
$M$ enters a loop with respect to the current state and the topmost pushdown
symbol while reading an $a$-block or a $b$-block.

First, consider the $a$-block. 
After reading the $a$-block, $M$ has $n$ symbols stored in the pushdown. 
In a sweep where the $\border$ symbol is read, $M$ must either read the 
first $b$-block or the second $b$-block entirely. 
Assume it reads the first $b$-block and consider the input 
$a^n b^k \border b^\ell$ with $k > n$. Then, after reading this block, the
pushdown store is empty, and the rest of the input must be processed using
only the finite control. As a result, the length of the second $b$-block can
be increased which is again a contradiction. The same 
argument applies analogously if the second $b$-block is read first.

Second, consider one of the $b$-blocks. Since $b\in\Sigma_R$,
$M$ runs through a state loop with empty pushdown while processing the block.
So the length of the block can be increased by the length of the loop without
changing the overall result of the computation, a contradiction.
 
\textbf{Case 4:}
The fourth case is that $a\in \Sigma_R$ and $b\in \Sigma_D$ and
$M$ enters a loop with respect to the current state and the topmost pushdown
symbol while reading an $a$-block or a $b$-block.
The argumentation in this case is symmetric to Case 3. Here
the roles of $a$ and $b$ (that is, their memberships in $\Sigma_D$ and
$\Sigma_R$) are switched. 

Finally, from the contradictions in all possible cases we conclude that $L$ is
not accepted by any $\nrdidpdawtl$.
\end{proof}

The additional power of nondeterministic input-driven pushdown automata with
translucent input letters versus their deterministic counterpart
is due to the fact that the translucency of input letters
allows different computation paths to treat the same input symbol differently,
thereby enabling different operations on the pushdown store. In a
nondeterministic setting, this flexibility permits branching into multiple
computational paths, each potentially using the pushdown store in a distinct
way. However, a deterministic machine cannot simulate these differing
operations simultaneously, which limits its expressive power.

\section{Returning versus Non-Returning}\label{sect:ret}

It is known that deterministic pushdown automata with translucent letters
working in the non-returning mode can accept even a non-semilinear 
language~\cite{kutrib:2024:dpdawtl:proc}. So a natural question is whether
the same still holds for the structurally weaker device which has
to obey the input-driven mode. The next theorem answers this question in the
affirmative. To this end, we tweak the witness language from~\cite{kutrib:2024:dpdawtl:proc}.
We define the non-semilinear language $L_{\textsl{nsl}}$ as
$$
L_{\textsl{nsl}} = \{\, a\,\dollar\border\, a^3\,\dollar\border^2\,
a^5\,\dollar\border^3 \cdots\, \dollar\border^{k-1}\, a^{2k-1}\, \dollar\cent^{k}\, a^{2k+1}\mid
k\geq 0\,\}.
$$

\begin{theorem}\label{theo:nrdidpdawtl-not-semilinear}
The language $L_{\textsl{nsl}}$ is accepted by some $\nrdidpdawtl$.
\end{theorem}

The language $L_{\textsl{nsl}}$ is accepted by some $\nrdidpdawtl$ and thus by
some $\nrnidpdawtl$. It is known that all languages accepted by $\dpdawtl$ are
semilinear~\cite{kutrib:2024:dpdawtl:proc}. So, the next question is whether 
this is still true when we trade nondeterminism for the input-driven property.
The following result has been shown in~\cite{nagy:2011:fsawtl} for finite 
automata with translucent letters and can be adapted for our purposes.

\begin{proposition}\label{prop:letter-equivalent}
From any given $\nidpdawtl$ $M$, an $\npda$ $M'$ can effectively be constructed, such that
$L(M')\subseteq L(M)$ and $L(M')$ is letter-equivalent to $L(M)$.
\end{proposition}

Proposition~\ref{prop:letter-equivalent} together with the well-known result that 
all context-free languages are semilinear~\cite{Parikh:1966:ocfl} implies that 
all languages in $\lfam(\nidpdawtl)\supset \lfam(\didpdawtl)$ are semilinear.
So, by Theorem~\ref{theo:nrdidpdawtl-not-semilinear} we have the following corollary. 

\begin{corollary}\label{cor:ret-in-nonret}
The family $\lfam(\didpdawtl)$ is properly included in 
$\lfam(\nrdidpdawtl)$, and the family $\lfam(\nidpdawtl)$ is properly included in 
$\lfam(\nrnidpdawtl)$.
\end{corollary}

The remaining two language families under consideration to be compared are
$\lfam(\nrdidpdawtl)$ and $\lfam(\nidpdawtl)$.

\begin{proposition}
The language families $\lfam(\nrdidpdawtl)$ and $\lfam(\nidpdawtl)$
are incomparable.
\end{proposition}

\begin{proof}
Theorem~\ref{theo:nrdidpdawtl-not-semilinear} provides a non-semilinear language accepted by
some $\nrdidpdawtl$ but not accepted by any $\nidpdawtl$. 

Conversely,
the proof of Theorem~\ref{theo:det-in-ndet} provides a
language that is accepted by some $\nidpdawtl$ but cannot be accepted by any
$\nrdidpdawtl$.
\end{proof}

\section{Closure under Boolean Operations}\label{sec:closure}

Often nondeterministic devices induce language families that are closed under
union but are not closed under intersection. This implies immediately the
non-closure under complementation. However, here the closure under union is
an open problem. Therefore, we show the non-closure under complementation directly.
A witness for the nondeterministic families $\lfam(\nrnidpdawtl)$ and
$\lfam(\nidpdawtl)$ is the language
$$
L_{\textsl{rep}}= \{\, b^n (\border b^n)^k \shuffle a^n \mid n\geq 1, k\geq 0\,\}.
$$

\begin{theorem}\label{theo:nondet-not-closed-complement}
The language families $\lfam(\nrnidpdawtl)$ and $\lfam(\nidpdawtl)$ are not
closed under complementation.
\end{theorem}

It is well known that the families of languages induced by deterministic
pushdown automata and real-time deterministic pushdown automata are closed
under complementation. The closure has also been derived for
$\dpdawtl$ and $\nrdpdawtl$~\cite{kutrib:2024:dpdawtl:proc}. Here we
complement these results by showing the closure for the deterministic language families
studied here.

\begin{proposition}\label{prop:det-closed-complement}
The language families $\lfam(\nrdidpdawtl)$ and $\lfam(\didpdawtl)$ are
closed under complementation.
\end{proposition}

Next, we show that both deterministic families are not closed under the remaining
Boolean operations union and intersection.

\begin{proposition}\label{prop:det-notclosed-union}
The language families $\lfam(\nrdidpdawtl)$ and $\lfam(\didpdawtl)$ are
neither closed under union nor under intersection.
\end{proposition}

\begin{proof}
We use the languages
$
 L_1= \{\, b^n\border b^m \shuffle a^n\mid m,n\geq 1\,\}
$
and
$
 L_2 = \{\, b^m\border b^n \shuffle a^n \mid m,n\geq 1\,\}
$ 
introduced in the proof of Theorem~\ref{theo:det-in-ndet}.
Each language can be accepted by a $\didpdawtl$ with signature $\Sigma_D=\{a\}$, $\Sigma_R=\{b\}$, and $\Sigma_N=\{\border\}$.
The rough idea for a $\didpdawtl$ $M_1$ accepting $L_1$ is to push 
in a first phase all $a$'s while symbols $b$ and $\border$ are translucent.
In a second phase, the first block of $b$'s is matched against the pushdown store and the second block of $b$'s is in
fact ignored since $M_1$ may pop from the empty pushdown only.

The rough idea for a $\didpdawtl$ $M_2$ accepting $L_2$ is to consume all
$b$'s up to and including the symbol $\border$ in a first phase. At the end of
this phase the pushdown store is empty. In a second phase, all symbols~$b$ are
translucent and the $a$'s are consumed and pushed. Finally, in a third phase,
the remaining~$b$'s from the input are matched against the pushdown store.

Since it is shown in Theorem~\ref{theo:det-in-ndet} that the union $L_1 \cup L_2$ is not accepted by any
$\nrdidpdawtl$, we obtain the non-closure under union for $\lfam(\didpdawtl)$ as well as for $\lfam(\nrdidpdawtl)$, 
even if the given automata have identical signatures. 
Since both families are closed under complementation by Proposition~\ref{prop:det-closed-complement},
we also obtain the non-closure under intersection for $\lfam(\didpdawtl)$ as well as for $\lfam(\nrdidpdawtl)$.
\end{proof}

For the nondeterministic families we already know the non-closure under complementation.
We now show that the nondeterministic returning class is not closed under intersection even with regular languages.
To this end, we use the result of Proposition~\ref{prop:letter-equivalent} stating that
from any given $\nidpdawtl$~$M$ we can effectively construct an $\npda$ $M'$ such that
$L(M')\subseteq L(M)$ and $L(M')$ is letter-equivalent to~$L(M)$.
Since the context-sensitive language $L_{abc} =\{\, a^n b^n c^n \mid n \ge 0 \,\}$ does not contain any 
letter-equivalent context-free sub-language, we can conclude that $L_{abc} \not\in \lfam(\nidpdawtl)$.
On the other hand, $L_{abc}=L_{a,b,c} \cap a^*b^*c^*$ with $L_{a,b,c} =\{\, w \in \{a,b,c\}^* \mid |w|_a=|w|_b=|w|_c \,\}$
and $L_{a,b,c}$ as well as the regular language $a^*b^*c^*$ can be accepted by $\nidpdawtl$ having
identical signatures. Hence, we obtain the following proposition.

\begin{proposition}\label{prop:ndet-ret-notclosed-union}
The language family $\lfam(\nidpdawtl)$ is neither closed under intersection
nor under intersection with regular languages.
\end{proposition}

It remains an open question whether or not the family $\lfam(\nidpdawtl)$ is closed under union and
whether or not the family $\lfam(\nrnidpdawtl)$ is closed under union or intersection.
Obviously, we obtain the closure under union for both families if the signatures are compatible.
However, we strongly conjecture non-closure in all other cases.

\section{Decidability Questions}\label{sect:deci}

In this section, we investigate decidability questions such as, for example, emptiness, finiteness, universality, inclusion,
equivalence, and regularity for our introduced input-driven variants of pushdown automata with translucent letters. 
These decidability questions have already been investigated for deterministic and nondeterministic finite automata with
translucent letters in~\cite{Nagy:2012:ocdsosdrawws,Nagy:2013:gdcdsosrawws1} where some partial results have been obtained. 
For returning deterministic and nondeterministic finite automata with translucent letters the questions of emptiness and finiteness 
are decidable. In addition, universality is decidable in the deterministic case. Inclusion is already undecidable for returning
deterministic finite automata with translucent letters. This negative result carries over to all other models in the returning and/or
nondeterministic case. For returning nondeterministic finite automata with translucent letters the questions of equivalence
and regularity are undecidable and carry over to non-returning nondeterministic finite automata with translucent letters as well.
Here, we study these questions for pushdown automata with translucent letters. We note that some of the undecidability results
obtained for finite automata with translucent letters carry over to pushdown automata with translucent letters. However, we show
here the non-semidecidability of the problems and, in addition, the non-semidecidability of universality for nondeterministic
input-driven pushdown automata in the returning and non-returning case. We start with decidable questions and
show that some questions are decidable for returning pushdown automata with translucent letters.

\begin{theorem}\label{theo:dec:univ}
For $\didpdawtl$ or $\dpdawtl$ as input,
the problems of testing emptiness, finiteness, and universality
are decidable.
For $\nidpdawtl$ as input, the problems of testing emptiness and finiteness
are decidable.
\end{theorem}

Next, we switch to undecidability results and we will show, in particular, the non-semidecidability of some problems.
To prove these results we use the technique of valid 
computations of Turing machines~\cite{Hartmanis:1967:cfltmc}. 
It suffices to consider deterministic Turing machines with one single
read-write head and one single tape whose space is fixed by the length of the
input, that is, so-called linear bounded automata ($\lba$).
Without loss of generality and for technical reasons, we assume that the $\lba$s
can halt only after an odd number of moves, accept by halting, and make
at least three moves. 
A valid computation is a string built from a sequence of configurations
passed through during an accepting computation. 

Let $Q$ be the state set of some $\lba$~$M$,
where $q_0$ is the initial state, 
$T \cap Q = \emptyset$ is the tape alphabet containing the endmarkers
$\underline{\leftend}$ and $\underline{\rightend}$, and $\Sigma\subset T$ is the input
alphabet.
A configuration of~$M$ can be written as a string of the form
$\underline{\leftend} T^*QT^*\underline{\rightend}$ such that, 
$\underline{\leftend} t_1t_2\cdots t_i q t_{i+1} \cdots t_n\underline{\rightend}$ is used to express
that $\underline{\leftend} t_1t_2\cdots t_n\underline{\rightend}$ is the tape inscription, $M$ is in state $q$, 
and is scanning tape symbol $t_{i+1}$.

The set of valid computations $\valc(M)$ is now defined to be 
the set of words having the form 
$
 w_0\border w_2\border \cdots \border w_{2m} c w^R_{2m+1} \border w^R_{2m-1} \border \cdots
\border w^R_3 \border w^R_1,
$
where $\border, c \notin T\cup Q$,
$w_i \in \underline{\leftend} T^*QT^*\underline{\rightend}$ are configurations of~$M$,
$w_0$ is an initial configuration from $\underline{\leftend}q_0\Sigma^*\underline{\rightend}$,
$w_{2m+1}$ is a halting, that is, an accepting configuration, 
and $w_{i+1}$ is the successor configuration of~$w_i$ for $0 \le i \le 2m$.
The set of invalid computations $\invalc(M)$ is the defined as the complement of $\valc(M)$
with respect to the alphabet $T \cup Q \cup \{\border, c\}$.

To accept the set $\invalc(M)$ by an $\nidpdawtl$ we make some modifications.
Let $h'$ be a mapping that maps every symbol from $T \cup Q \cup \{\border\}$ to its primed version.
Similarly, let $h''$ be a mapping that maps every symbol from $T \cup Q \cup \{\border\}$ to its double-primed version.
We then define the set of valid computations $\valc'(M)$ to be 
the set of words of the form 
$
 w_0\border w_2\border \cdots \border w_{2m} c \left(h'(w^R_{2m+1} \border) \shuffle h''(w^R_{2m-1} \border \cdots
\border w^R_3 \border w^R_1)\right),
$
where $w_0\border w_2\border \cdots \border w_{2m} c w^R_{2m+1} \border w^R_{2m-1} \border \cdots \border w^R_3 \border w^R_1$
belongs to $\valc(M)$. The set of invalid computations $\invalc\,'(M)$ is then defined as the complement of $\valc'(M)$.

\begin{lemma}\label{lemma:invalc}
Let $M$ be an $\lba$. Then, an $\nidpdawtl$ accepting $\invalc\,'(M)$ can effectively be constructed.
\end{lemma}

\begin{proof}
We sketch the construction of an $\nidpdawtl$ $M'$ accepting $\invalc\,'(M)$ whose signature is defined as
$\Sigma_D=T \cup Q \cup \{\border\}$, $\Sigma_N=\{c\}$, and $\Sigma_R=T' \cup T'' \cup Q' \cup Q'' \cup \{\border',\border''\}$.
To check whether an input $x$ belongs to $\invalc\,'(M)$, $M'$ guesses and tests one of the following four possibilities.
\begin{enumerate}
\item $x$ has the wrong format to belong to $\valc'(M)$.
\item $x$ has the correct format, but the number of configurations to the left of $c$ is different from the number of configurations
to the right of $c$.
\item $x=w_0\border w_2\border \cdots \border w_{2m} c \left(h'(w^R_{2m+1} \border) \shuffle h''(w^R_{2m-1} \border \cdots
\border w^R_3 \border w^R_1)\right)$, but $w_{2i+1}$ is not the successor configuration of $w_{2i}$ for some $0 \le i \le m$.
\item $x=w_0\border w_2\border \cdots \border w_{2m} c \left(h'(w^R_{2m+1} \border) \shuffle h''(w^R_{2m-1} \border \cdots
\border w^R_3 \border w^R_1)\right)$, but $w_{2i+2}$ is not the successor configuration of $w_{2i+1}$ for some $0 \le i \le m-1$.
\end{enumerate}
The first possibility can be tested by a finite automaton and, hence, by~$M'$ disregarding the actions on the pushdown store.
For the second possibility $M'$ acts as follows: it reads the input up to the middle marker~$c$ and pushes the input as it is on
the pushdown store. After reading the marker~$c$, $M'$ makes all symbols from $T'' \cup Q'' \cup \{\border''\}$ translucent and
pops for every input symbol from $T' \cup Q' \cup \{\border'\}$ a symbol from the pushdown store taking care that $\border'$ in the
input is only matched against $\border$ on the pushdown store. If an error is encountered in this phase, the input is accepted.
If $M'$ sees the endmarker, $M'$ reads the remaining input and pops for every input symbol from $T'' \cup Q'' \cup \{\border''\}$ 
a symbol from the pushdown store taking again care that $\border''$ in the input is only matched against $\border$ on the pushdown store. 
If an error is encountered in this phase, the input is accepted as well. If the input is read completely and the pushdown store is not
empty, or the pushdown store gets empty before the input is read completely, $M'$ accepts as well and rejects in all other cases.

\begin{table}[!b]
\begin{center}
\renewcommand{\arraystretch}{1.2}\setlength{\tabcolsep}{6pt}
\begin{tabular}{|c|c|c|c|c|c|c|}
\hline
\rowcolor[HTML]{EFEFEF} 
  Automata Family & {$\emptyset$} & {FIN} & ${\Sigma^*}$ &
  {$\subseteq$} & $=$ & REG\\
\hline\hline
\cellcolor[HTML]{EFEFEF}$\dpdawtl$   & \cyes & \cyes & \cyes &  \cno & \copen &\copen\\
\cellcolor[HTML]{EFEFEF}$\didpdawtl$   & \cyes & \cyes & \cyes & \cno & {\copen} &\copen\\
\cellcolor[HTML]{EFEFEF}$\nidpdawtl$   & \cyes & \cyes & \cnos & \cnos & {\cnos} &\cnos\\
\cellcolor[HTML]{EFEFEF}$\nrdpdawtl$   & \copen & \copen & \copen & \cno & {\copen} &\copen\\
\cellcolor[HTML]{EFEFEF}$\nrdidpdawtl$   & \copen & \copen & \copen & \cno & \copen &\copen\\
\cellcolor[HTML]{EFEFEF}$\nrnidpdawtl$   & \copen & \copen & \cnos & \cnos & {\cnos} &\cnos\\
\hline
\end{tabular}
\end{center}
\caption{A summary of decidability questions for the language families discussed in this paper.
The undecidable questions derived from finite automata with translucent letters are marked with~`\cno', whereas the non-semidecidable questions obtained
in this paper are marked with~`\cnos'.}
\label{tab:closure}
\end{table}

To test the third possibility $M'$ reads the input up to the middle marker~$c$ and pushes the input as it is on the pushdown store.
Additionally, at some point of time $M'$ guesses the index~$i$. Then, $M'$ pushes configuration~$w_{2i}$ with suitably marked symbols 
on the pushdown store and~$M'$ remembers the last three symbols read in its finite control until the state symbol of configuration~$w_{2i}$ 
is the middle one of these three. After reading the middle marker~$c$ the task is to identify configuration~$w_{2i+1}$ in the input
and to check that $w_{2i+1}$ is not the successor configuration of~$w_{2i}$. If the suitably marked configuration on the pushdown store
is the topmost one after reading~$c$, $M'$ makes all symbols from $T'' \cup Q'' \cup \{\border''\}$ translucent and
pops for every input symbol from $T' \cup Q' \cup \{\border'\}$ a symbol from the pushdown store verifying
that the current configuration is \emph{not} the reversal of the successor configuration of the configuration stored in the pushdown store. 
Both configurations differ only locally at the state symbol. But from the information remembered in the finite control, the differences 
can be computed and verified. If an error is encountered, the input is accepted and otherwise rejected.
If the suitably marked configuration on the pushdown store
is not the topmost one after reading~$c$, then $M'$ makes all symbols from $T'' \cup Q'' \cup \{\border''\}$ translucent and
pops for every input symbol from $T' \cup Q' \cup \{\border'\}$ a symbol from the pushdown store checking the correct length and format
as in the test of the second possibility. After this phase handling inputs from $T' \cup Q' \cup \{\border'\}$,
$M'$ reads the remaining input and pops for every input symbol from $T'' \cup Q'' \cup \{\border''\}$ a symbol from the pushdown store 
checking again the correct length and format as in the test of the second
possibility until the suitably marked symbols appear on
the pushdown store. In this case, $M'$ pops for every input symbol from $T'' \cup Q'' \cup \{\border''\}$ a symbol from the pushdown store 
verifying that the current configuration is \emph{not} the reversal of the successor configuration of the configuration stored in the pushdown store. 
Again, this can be computed and verified due to the information remembered in the finite control, since both configurations differ only 
locally at the state symbol. If an error is encountered, the input is accepted and otherwise rejected.

The idea to test the fourth possibility is in a first phase identical to the third possibility: 
$M'$ reads the input up to the middle marker~$c$ and pushes the input as it is on the pushdown store. Additionally,~$M'$ pushes 
configuration~$w_{2i+2}$ with suitably marked symbols and remembers the last three symbols read in its finite control 
until the state symbol of configuration~$w_{2i+2}$ is the middle one of these three.
After reading the middle marker~$c$ the task is to identify configuration~$w_{2i+1}$ in the input
and to check that~$w_{2i+1}$ is not the successor configuration of~$w_{2i+2}$. To this end, $M'$ makes all symbols from $T' \cup Q' \cup \{\border'\}$
translucent and pops for every input symbol from $T'' \cup Q'' \cup \{\border''\}$ a symbol from the pushdown store checking the correct length and format
as in the test of the second possibility until the suitably marked symbols appear on the pushdown store. 
In this case, $M'$ pops for every input symbol from $T'' \cup Q'' \cup \{\border''\}$ a symbol from the pushdown store 
verifying that the reversal of the successor configuration of the current configuration is \emph{not} the configuration stored in the pushdown store. 
Again, this can be computed and verified due to the information remembered in the finite control, since both configurations differ only 
locally at the state symbol. If an error is encountered, the input is accepted and otherwise rejected.
We note that it is implicitly detected by possibilities 3 and 4 if all configurations do not have the same length.
This completes the construction of the $\nidpdawtl$ $M'$ accepting $\invalc\,'(M)$.
\end{proof}

The fact that $\nidpdawtl$ accept the set of invalid computations of an $\lba$ is sufficient to obtain the next non-semidecidability
results.

\begin{theorem}\label{theo:undec:univ}
For $\nidpdawtl$ or $\nrnidpdawtl$ as input, the problems of testing universality, inclusion, equivalence, and regularity
are not semidecidable.
\end{theorem}

\bibliographystyle{eptcs}
\bibliography{paper}  

\def\lastmodified{3.4.2025}\def\id#1{#1}
\begin{thebibliography}{10}
\providecommand{\bibitemdeclare}[2]{}
\providecommand{\surnamestart}{}
\providecommand{\surnameend}{}
\providecommand{\urlprefix}{Available at }
\providecommand{\url}[1]{\texttt{#1}}
\providecommand{\href}[2]{\texttt{#2}}
\providecommand{\urlalt}[2]{\href{#1}{#2}}
\providecommand{\doi}[1]{doi:\urlalt{https://doi.org/#1}{#1}}
\providecommand{\eprint}[1]{arXiv:\urlalt{https://arxiv.org/abs/#1}{#1}}
\providecommand{\bibinfo}[2]{#2}

\bibitemdeclare{article}{beier:2022:nrowjfa}
\bibitem{beier:2022:nrowjfa}
\bibinfo{author}{Simon \surnamestart Beier\surnameend} \&
  \bibinfo{author}{Markus \surnamestart Holzer\surnameend}
  (\bibinfo{year}{2022}): \emph{\bibinfo{title}{Nondeterministic right one-way
  jumping finite automata}}.
\newblock {\slshape \bibinfo{journal}{Inform. Comput.}} \bibinfo{volume}{284},
  p. \bibinfo{pages}{104687}, \doi{10.1016/J.IC.2021.104687}.

\bibitemdeclare{article}{bensch:2009:irdnfa}
\bibitem{bensch:2009:irdnfa}
\bibinfo{author}{Suna \surnamestart Bensch\surnameend},
  \bibinfo{author}{Henning \surnamestart Bordihn\surnameend},
  \bibinfo{author}{Markus \surnamestart Holzer\surnameend} \&
  \bibinfo{author}{Martin \surnamestart Kutrib\surnameend}
  (\bibinfo{year}{2009}): \emph{\bibinfo{title}{On input-revolving
  deterministic and nondeterministic finite automata}}.
\newblock {\slshape \bibinfo{journal}{Inform. Comput.}} \bibinfo{volume}{207},
  pp. \bibinfo{pages}{1140--1155}, \doi{10.1016/j.ic.2009.03.002}.

\bibitemdeclare{incollection}{braunmuehl:1985:idlrlog-coll}
\bibitem{braunmuehl:1985:idlrlog-coll}
\bibinfo{author}{Burchard \surnamestart von Braunm{\"u}hl\surnameend} \&
  \bibinfo{author}{Rutger \surnamestart Verbeek\surnameend}
  (\bibinfo{year}{1985}): \emph{\bibinfo{title}{Input-Driven Languages are
  Recognized in {$\log n$} Space}}.
\newblock In \bibinfo{editor}{Marek \surnamestart Karpinski\surnameend} \&
  \bibinfo{editor}{Jan \surnamestart van Leeuwen\surnameend}, editors:
  {\slshape \bibinfo{booktitle}{Topics in the Theory of Computation}},
  {\slshape \bibinfo{series}{Mathematics Studies}} \bibinfo{volume}{102},
  \bibinfo{publisher}{North-Holland}, \bibinfo{address}{Amsterdam}, pp.
  \bibinfo{pages}{1--19}, \doi{10.1016/S0304-0208(08)73072-X}.

\bibitemdeclare{article}{chigahara:2016:owjfa}
\bibitem{chigahara:2016:owjfa}
\bibinfo{author}{Hiroyuki \surnamestart Chigahara\surnameend},
  \bibinfo{author}{Szil{\'{a}}rd~Zsolt \surnamestart Fazekas\surnameend} \&
  \bibinfo{author}{Akihiro \surnamestart Yamamura\surnameend}
  (\bibinfo{year}{2016}): \emph{\bibinfo{title}{One-Way Jumping Finite
  Automata}}.
\newblock {\slshape \bibinfo{journal}{Int. J. Found. Comput. Sci.}}
  \bibinfo{volume}{27}, pp. \bibinfo{pages}{391--405},
  \doi{10.1142/S0129054116400165}.

\bibitemdeclare{article}{Hartmanis:1967:cfltmc}
\bibitem{Hartmanis:1967:cfltmc}
\bibinfo{author}{Juris \surnamestart Hartmanis\surnameend}
  (\bibinfo{year}{1967}): \emph{\bibinfo{title}{Context-free languages and
  {T}uring machine computations}}.
\newblock {\slshape \bibinfo{journal}{Proc. Symposia in Applied Mathematics}}
  \bibinfo{volume}{19}, pp. \bibinfo{pages}{42--51},
  \doi{10.1090/psapm/019/0235938}.

\bibitemdeclare{book}{Hopcroft:1979:itatlc:book}
\bibitem{Hopcroft:1979:itatlc:book}
\bibinfo{author}{John~E. \surnamestart Hopcroft\surnameend} \&
  \bibinfo{author}{Jeffrey~D. \surnamestart Ullman\surnameend}
  (\bibinfo{year}{1979}): \emph{\bibinfo{title}{Introduction to Automata
  Theory, Languages, and Computation}}.
\newblock \bibinfo{publisher}{Addison-Wesley}, \bibinfo{address}{Reading,
  Massachusetts}.

\bibitemdeclare{inproceedings}{Jancar:1995:ra:proc}
\bibitem{Jancar:1995:ra:proc}
\bibinfo{author}{Petr \surnamestart Jan{\v{c}}ar\surnameend},
  \bibinfo{author}{Franti{\v{s}}ek \surnamestart Mr{\'a}z\surnameend},
  \bibinfo{author}{Martin \surnamestart Pl{\'a}tek\surnameend} \&
  \bibinfo{author}{J{\"o}rg \surnamestart Vogel\surnameend}
  (\bibinfo{year}{1995}): \emph{\bibinfo{title}{Restarting automata}}.
\newblock In \bibinfo{editor}{Horst \surnamestart Reichel\surnameend}, editor:
  {\slshape \bibinfo{booktitle}{Fundamentals of Computation Theory (FCT
  1995)}}, {\slshape \bibinfo{series}{LNCS}} \bibinfo{volume}{965},
  \bibinfo{publisher}{Springer}, pp. \bibinfo{pages}{283--292},
  \doi{10.1007/3-540-60249-6\_60}.

\bibitemdeclare{inproceedings}{kutrib:2025:twfawtil:proc}
\bibitem{kutrib:2025:twfawtil:proc}
\bibinfo{author}{Martin \surnamestart Kutrib\surnameend},
  \bibinfo{author}{Andreas \surnamestart Malcher\surnameend},
  \bibinfo{author}{Carlo \surnamestart Mereghetti\surnameend} \&
  \bibinfo{author}{Beatrice \surnamestart Palano\surnameend}
  (\bibinfo{year}{2025}): \emph{\bibinfo{title}{Two-Way Finite Automata with
  Translucent Input Letters}}.
\newblock In \bibinfo{editor}{Andreas \surnamestart Malcher\surnameend} \&
  \bibinfo{editor}{Luca \surnamestart Prigioniero\surnameend}, editors:
  {\slshape \bibinfo{booktitle}{Descriptional Complexity of Formal Systems
  (DCFS 2025)}}, {\slshape \bibinfo{series}{LNCS}} \bibinfo{volume}{15759},
  \bibinfo{publisher}{Springer}, pp. \bibinfo{pages}{151--165},
  \doi{10.1007/978-3-031-97100-6\_11}.

\bibitemdeclare{inproceedings}{kutrib:2024:dpdawtl:proc}
\bibitem{kutrib:2024:dpdawtl:proc}
\bibinfo{author}{Martin \surnamestart Kutrib\surnameend},
  \bibinfo{author}{Andreas \surnamestart Malcher\surnameend},
  \bibinfo{author}{Carlo \surnamestart Mereghetti\surnameend},
  \bibinfo{author}{Beatrice \surnamestart Palano\surnameend},
  \bibinfo{author}{Priscilla \surnamestart Raucci\surnameend} \&
  \bibinfo{author}{Matthias \surnamestart Wendlandt\surnameend}
  (\bibinfo{year}{2024}): \emph{\bibinfo{title}{Deterministic Pushdown Automata
  with Translucent Input Letters}}.
\newblock In \bibinfo{editor}{Joel~D. \surnamestart Day\surnameend} \&
  \bibinfo{editor}{Florin \surnamestart Manea\surnameend}, editors: {\slshape
  \bibinfo{booktitle}{Developments in Language Theory (DLT 2024)}}, {\slshape
  \bibinfo{series}{LNCS}} \bibinfo{volume}{14791},
  \bibinfo{publisher}{Springer}, pp. \bibinfo{pages}{203--217},
  \doi{10.1007/978-3-031-66159-4\_15}.

\bibitemdeclare{inproceedings}{Kutrib:2024:opolabdpawtil}
\bibitem{Kutrib:2024:opolabdpawtil}
\bibinfo{author}{Martin \surnamestart Kutrib\surnameend},
  \bibinfo{author}{Andreas \surnamestart Malcher\surnameend},
  \bibinfo{author}{Carlo \surnamestart Mereghetti\surnameend},
  \bibinfo{author}{Beatrice \surnamestart Palano\surnameend},
  \bibinfo{author}{Priscilla \surnamestart Raucci\surnameend} \&
  \bibinfo{author}{Matthias \surnamestart Wendlandt\surnameend}
  (\bibinfo{year}{2024}): \emph{\bibinfo{title}{On Properties of Languages
  Accepted by Deterministic Pushdown Automata with Translucent Input Letters}}.
\newblock In \bibinfo{editor}{Szil{\'{a}}rd~Zsolt \surnamestart
  Fazekas\surnameend}, editor: {\slshape \bibinfo{booktitle}{Implementation and
  Application of Automata ({CIAA} 2024)}}, {\slshape \bibinfo{series}{LNCS}}
  \bibinfo{volume}{15015}, \bibinfo{publisher}{Springer}, pp.
  \bibinfo{pages}{208--220}, \doi{10.1007/978-3-031-71112-1\_15}.

\bibitemdeclare{article}{meduna:2012:jfa}
\bibitem{meduna:2012:jfa}
\bibinfo{author}{Alexander \surnamestart Meduna\surnameend} \&
  \bibinfo{author}{Petr \surnamestart Zemek\surnameend} (\bibinfo{year}{2012}):
  \emph{\bibinfo{title}{Jumping finite automata}}.
\newblock {\slshape \bibinfo{journal}{Int. J. Found. Comput. Sci.}}
  \bibinfo{volume}{23}, pp. \bibinfo{pages}{1555--1578},
  \doi{10.1142/S0129054112500244}.

\bibitemdeclare{inproceedings}{mehlhorn:1980:pmradcflr}
\bibitem{mehlhorn:1980:pmradcflr}
\bibinfo{author}{Kurt \surnamestart Mehlhorn\surnameend}
  (\bibinfo{year}{1980}): \emph{\bibinfo{title}{Pebbling Mountain Ranges and
  its Application of {DCFL}-Recognition}}.
\newblock In \bibinfo{editor}{J.~W. \surnamestart de~Bakker\surnameend} \&
  \bibinfo{editor}{Jan \surnamestart van Leeuwen\surnameend}, editors:
  {\slshape \bibinfo{booktitle}{International Colloquium on Automata, Languages
  and Programming (ICALP 1980)}}, {\slshape
  \bibinfo{series}{LNCS}}~\bibinfo{volume}{85}, \bibinfo{publisher}{Springer},
  pp. \bibinfo{pages}{422--435}, \doi{10.1007/3-540-10003-2\_89}.

\bibitemdeclare{article}{mraz:2023:nrfawtl}
\bibitem{mraz:2023:nrfawtl}
\bibinfo{author}{Frantisek \surnamestart Mr{\'{a}}z\surnameend} \&
  \bibinfo{author}{Friedrich \surnamestart Otto\surnameend}
  (\bibinfo{year}{2023}): \emph{\bibinfo{title}{Non-returning deterministic and
  nondeterministic finite automata with translucent letters}}.
\newblock {\slshape \bibinfo{journal}{{RAIRO} Theor. Informatics Appl.}}
  \bibinfo{volume}{57}, p.~\bibinfo{pages}{8}, \doi{10.1051/ITA/2023009}.

\bibitemdeclare{article}{nagy:2011:cdsosdragbaeps}
\bibitem{nagy:2011:cdsosdragbaeps}
\bibinfo{author}{Benedek \surnamestart Nagy\surnameend} \&
  \bibinfo{author}{Friedrich \surnamestart Otto\surnameend}
  (\bibinfo{year}{2011}): \emph{\bibinfo{title}{{CD}-systems of stateless
  deterministic {R}(1)-automata governed by an external pushdown store}}.
\newblock {\slshape \bibinfo{journal}{{RAIRO} Theor. Informatics Appl.}}
  \bibinfo{volume}{45}, pp. \bibinfo{pages}{413--448},
  \doi{10.1051/ITA/2011123}.

\bibitemdeclare{inproceedings}{nagy:2011:fsawtl}
\bibitem{nagy:2011:fsawtl}
\bibinfo{author}{Benedek \surnamestart Nagy\surnameend} \&
  \bibinfo{author}{Friedrich \surnamestart Otto\surnameend}
  (\bibinfo{year}{2011}): \emph{\bibinfo{title}{Finite-state Acceptors with
  Translucent Letters}}.
\newblock In \bibinfo{editor}{G.~\surnamestart Bel-Enguix\surnameend},
  \bibinfo{editor}{V.~\surnamestart Dahl\surnameend} \& \bibinfo{editor}{A.O.
  \surnamestart De~La~Puente\surnameend}, editors: {\slshape
  \bibinfo{booktitle}{International Workshop on {AI} Methods for
  Interdisciplinary Research in Language and Biology (BILC 2011)}},
  \bibinfo{publisher}{SciTePress}, pp. \bibinfo{pages}{3--13},
  \doi{10.5220/0003272500030013}.

\bibitemdeclare{article}{Nagy:2012:ocdsosdrawws}
\bibitem{Nagy:2012:ocdsosdrawws}
\bibinfo{author}{Benedek \surnamestart Nagy\surnameend} \&
  \bibinfo{author}{Friedrich \surnamestart Otto\surnameend}
  (\bibinfo{year}{2012}): \emph{\bibinfo{title}{On {CD}-systems of stateless
  deterministic {R}-automata with window size one}}.
\newblock {\slshape \bibinfo{journal}{J. Comput. Syst. Sci.}}
  \bibinfo{volume}{78}, pp. \bibinfo{pages}{780--806},
  \doi{10.1016/J.JCSS.2011.12.009}.

\bibitemdeclare{article}{Nagy:2013:dpcdsosdra}
\bibitem{Nagy:2013:dpcdsosdra}
\bibinfo{author}{Benedek \surnamestart Nagy\surnameend} \&
  \bibinfo{author}{Friedrich \surnamestart Otto\surnameend}
  (\bibinfo{year}{2013}): \emph{\bibinfo{title}{Deterministic
  pushdown-{CD}-systems of stateless deterministic {R(1)}-automata}}.
\newblock {\slshape \bibinfo{journal}{Acta Inform.}} \bibinfo{volume}{50}, pp.
  \bibinfo{pages}{229--255}, \doi{10.1007/S00236-012-0175-X}.

\bibitemdeclare{article}{Nagy:2013:gdcdsosrawws1}
\bibitem{Nagy:2013:gdcdsosrawws1}
\bibinfo{author}{Benedek \surnamestart Nagy\surnameend} \&
  \bibinfo{author}{Friedrich \surnamestart Otto\surnameend}
  (\bibinfo{year}{2013}): \emph{\bibinfo{title}{Globally deterministic
  CD-systems of stateless R-automata with window size 1}}.
\newblock {\slshape \bibinfo{journal}{Int. J. Comput. Math.}}
  \bibinfo{volume}{90}(\bibinfo{number}{6}), pp. \bibinfo{pages}{1254--1277},
  \doi{10.1080/00207160.2012.688820}.

\bibitemdeclare{inproceedings}{Nagy:2024:fawsotw}
\bibitem{Nagy:2024:fawsotw}
\bibinfo{author}{Benedek \surnamestart Nagy\surnameend} \&
  \bibinfo{author}{Friedrich \surnamestart Otto\surnameend}
  (\bibinfo{year}{2024}): \emph{\bibinfo{title}{Finite Automata with Sets of
  Translucent Words}}.
\newblock In \bibinfo{editor}{Joel~D. \surnamestart Day\surnameend} \&
  \bibinfo{editor}{Florin \surnamestart Manea\surnameend}, editors: {\slshape
  \bibinfo{booktitle}{Developments in Language Theory (DLT 2024)}}, {\slshape
  \bibinfo{series}{LNCS}} \bibinfo{volume}{14791},
  \bibinfo{publisher}{Springer}, pp. \bibinfo{pages}{236--251},
  \doi{10.1007/978-3-031-66159-4\_17}.

\bibitemdeclare{article}{okhotin:2014:cidpda}
\bibitem{okhotin:2014:cidpda}
\bibinfo{author}{Alexander \surnamestart Okhotin\surnameend} \&
  \bibinfo{author}{Kai \surnamestart Salomaa\surnameend}
  (\bibinfo{year}{2014}): \emph{\bibinfo{title}{Complexity of input-driven
  pushdown automata}}.
\newblock {\slshape \bibinfo{journal}{SIGACT News}} \bibinfo{volume}{45}, pp.
  \bibinfo{pages}{47--67}, \doi{10.1145/2636805.2636821}.

\bibitemdeclare{inproceedings}{Otto:2015:ovptl}
\bibitem{Otto:2015:ovptl}
\bibinfo{author}{Friedrich \surnamestart Otto\surnameend}
  (\bibinfo{year}{2015}): \emph{\bibinfo{title}{On Visibly Pushdown Trace
  Languages}}.
\newblock In \bibinfo{editor}{Giuseppe~F. \surnamestart Italiano\surnameend},
  \bibinfo{editor}{Tiziana \surnamestart Margaria{-}Steffen\surnameend},
  \bibinfo{editor}{Jaroslav \surnamestart Pokorn{\'{y}}\surnameend},
  \bibinfo{editor}{Jean{-}Jacques \surnamestart Quisquater\surnameend} \&
  \bibinfo{editor}{Roger \surnamestart Wattenhofer\surnameend}, editors:
  {\slshape \bibinfo{booktitle}{{SOFSEM} 2015}}, {\slshape
  \bibinfo{series}{LNCS}} \bibinfo{volume}{8939},
  \bibinfo{publisher}{Springer}, pp. \bibinfo{pages}{389--400},
  \doi{10.1007/978-3-662-46078-8\_32}.

\bibitemdeclare{inproceedings}{otto:2023:asoawtl:proc}
\bibitem{otto:2023:asoawtl:proc}
\bibinfo{author}{Friedrich \surnamestart Otto\surnameend}
  (\bibinfo{year}{2023}): \emph{\bibinfo{title}{A Survey on automata with
  translucent letters}}.
\newblock In \bibinfo{editor}{Benedek \surnamestart Nagy\surnameend}, editor:
  {\slshape \bibinfo{booktitle}{Implementation and Application of Automata
  ({CIAA} 2023)}}, {\slshape \bibinfo{series}{LNCS}} \bibinfo{volume}{14151},
  \bibinfo{publisher}{Springer}, pp. \bibinfo{pages}{21--50},
  \doi{10.1007/978-3-031-40247-0\_2}.

\bibitemdeclare{book}{otto:2025:book}
\bibitem{otto:2025:book}
\bibinfo{author}{Friedrich \surnamestart Otto\surnameend}
  (\bibinfo{year}{2025}): \emph{\bibinfo{title}{Restarting Automata}}.
\newblock \bibinfo{publisher}{Springer}, \bibinfo{address}{Cham, Switzerland},
  \doi{10.1007/978-3-031-78701-0}.

\bibitemdeclare{article}{Parikh:1966:ocfl}
\bibitem{Parikh:1966:ocfl}
\bibinfo{author}{Rohit~J. \surnamestart Parikh\surnameend}
  (\bibinfo{year}{1966}): \emph{\bibinfo{title}{On Context-Free Languages}}.
\newblock {\slshape \bibinfo{journal}{J. ACM}} \bibinfo{volume}{13}, pp.
  \bibinfo{pages}{570--581}, \doi{10.1145/321356.321364}.

\end{thebibliography}

\end{document}